\documentclass[reqno,11pt]{amsart}
\usepackage{amsmath,amsthm}
\usepackage{amsfonts,amssymb}
\usepackage{xcolor}

\newtheorem{Theorem}{Theorem}[section]

\newtheorem{Proposition}[Theorem]{Proposition}
\newtheorem{Corollary}[Theorem]{Corollary}

\theoremstyle{definition}

\numberwithin{equation}{section}

\newcommand{\Cset}{\mathbb{C}}

\newcommand{\Nset}{\mathbb{N}}

\newcommand{\Vset}{\mathbb{V}}
\newcommand{\Vseto}{{\mathbb{V}}^{\text{o}}}
\newcommand{\Vsete}{{\mathbb{V}}^{\text{e}}}

\newcommand{\bfx}{\mathbf{x}}
\newcommand{\bft}{\mathbf{\theta}}

\newcommand{\la}{{\lambda}}

\newcommand{\pd}{\partial}

\newcommand{\cE}{{\mathcal E}}

\newcommand{\cO}{{\mathcal O}}

\newcommand{\osp}{\mathfrak{osp}(1|2)}

\begin{document}

\title{Superspace realizations of the Bannai--Ito algebra}

\author[N.~Cramp\'e]{Nicolas Cramp\'e}
\address{NC, Institut Denis-Poisson CNRS/UMR 7013 - Universit\'e de Tours - Universit\'e d'Orl\'eans,
Parc de Grandmont, 37200 Tours, France}
\email{crampe1977@gmail.com}

\author[H.~De~Bie]{Hendrik~De~Bie}
\address{HDB, Clifford research group, Department of electronics and information systems, Faculty of engineering and architecture, Ghent university,
Krijgslaan 281-S8, 9000, Gent (Belgium)}
\email{Hendrik.DeBie@UGent.be}

\author[P.~Iliev]{Plamen~Iliev}
\address{PI, School of Mathematics, Georgia Institute of Technology,
Atlanta, GA 30332--0160, USA}
\email{iliev@math.gatech.edu}

\author[L.~Vinet]{Luc Vinet}
\address{LV, IVADO and Centre de Recherches Math\'ematiques, Universit\'e de Montr\'eal, P.O. Box 6128, Centre-ville Station, Montr\'eal, QC H3C 3J7, Canada}
\email{vinet@crm.umontreal.ca}

\date{\today}

\maketitle

\begin{abstract}
    A model of the Bannai--Ito algebra in a superspace is introduced. It is obtained from the three-fold tensor product of the basic realization of the Lie superalgebra  $\osp$ in terms of operators in one continuous and one Grassmanian variable. The basis vectors of the resulting Bannai--Ito algebra module involve Jacobi polynomials.
\end{abstract}

\section{Introduction}
This paper offers a realization of the Bannai--Ito algebra in superspace. Models of algebras with operators acting on spaces of functions are especially relevant in mathematical physics and allow also to enrich the understanding of special functions \cite{miller1968lie}. This is the basic motivation here.
The Bannai-Ito algebra \cite{de2015bannai} encodes the bispectral properties of the eponymous polynomials which were first introduced in the context of P- and Q- polynomial association schemes \cite{bannai2021algebraic}. Their characterization as eigenfunctions of shift operators of Dunkl type was given in \cite{tsujimoto2012dunkl}. 

The Bannai--Ito polynomials have explicit expressions as combinations of two Racah polynomials \cite{koekoek2010hypergeometric}. Moreover, the Racah algebra \cite{genest2014racah} that is abstracted from the relations verified by the difference and recurrence operators corresponding to these last polynomials can be embedded in the Bannai--Ito algebra \cite{genest2015embeddings} much like the even Lie algebra part of a superalgebra can be obtained from quadratic expressions in the odd generators.

As a matter of fact, it was also observed that the Racah algebra can be directly embedded in the universal enveloping algebra $\mathcal{U}(\mathfrak{sl}(2))$ of $\mathfrak{sl}(2)$ \cite{gao2013classification}. Calling upon the Bargmann realization of $\mathfrak{sl}(2)$ in terms of differential operators in one variable \cite{perelomov2012generalized}, this last embedding thus immediately provides a realization of the Racah algebra in terms of differential operators of the hypergeometric type.

Given that the Racah polynomials arise in the $6-j$ recoupling coefficients of the $\mathfrak{sl}(2)$ representations, it is not surprising that a centrally extended version of the Racah algebra arises by taking as generators the intermediate Casimir in a three-fold product of representations \cite{genest2014superintegrability}. Here, the Bargmann realization produces a model of the Racah algebra in terms of differential operators in three variables. The dimensional reduction of this three-variable realization through separation of variables has been examined in 
\cite{genest2013equitable} and shown to yield the one-variable model of the Racah algebra mentioned before. It has been shown in this spirit how the higher rank Racah algebras \cite{de2017higher} constructed in this fashion using the $n$ variables associated to an n-fold product of irreducible $\mathfrak{sl}(2)$ representations can be reduced to a model involving differential operators in only $n-2$ variables \cite{de2019bargmann}. 

The Bannai-Ito algebra can also be viewed as the commutant of a diagonal embedding of an algebra module in a three-fold tensor product of representations. In this case the underlying algebraic structure is the Lie superalgebra $\mathfrak{osp}(1|2)$ and it is hence understood that the Bannai--Ito polynomials are essentially the Racah coefficients of this superalgebra \cite{genest2014bannai}. Now one extension to $\mathfrak{osp}(1|2)$ of the Bargmann model of $\mathfrak{sl}(2)$ involves thickening the one-dimensional space by the addition of a real Grassmann variable; this leads to using operators acting on functions defined on the resulting superspace \cite{dewitt1992supermanifolds}. Considering the tensor product of three irreducible (discrete series) representations of $\osp$ \cite{frappat1996dictionary} in this picture therefore brings three continuous variables and three Grasmannian ones. In view of the corresponding studies of the Racah realizations mentioned before, natural questions we address in this paper are: What is the Bannai--Ito algebra model that the irreducible decomposition of this tensor product representation brings? To how many variables does this reduction yield? What are the special functions involved? 

The remainder of the paper is organized as follows. The fundamental superspace model of $\osp$ is recalled in Section 2. The construction of the Bannai--Ito algebra generated by the intermediate Casimir operators of three copies of this $\osp$ realization is described in Section 3. This module decomposes into an even and an odd part each with four components. The dimensional reduction that occurs upon requesting that the representation space be bounded from below is discussed in Section 4. It will be shown that the kernel of the total lowering operator is such that its even and odd part only have two components. The action of the intermediate Casimir operators in each of these 4 parts will be given in Section 5. The construction of the irreducible representations is carried out by considering in turn the even subspace in Section 6 and the odd one in Section 7. In each case, first, the concrete representation basis will be obtained by diagonalizing the total Casimir operator and one of the intermediate Casimir elements and second, the tridiagonal action of the other independent intermediate Casimir element will be computed. The Jacobi polynomials will be seen to appear in the expression of these basis elements. The concluding section will offer summary, remarks and outlook.

\section{The fundamental realization of $\osp$ on superspace}
Consider the operators
\begin{align*}
A_{-} &=\theta\pd_{x}+\pd_{\theta},\\
A_{0} &=2x\pd_{x}+\theta\pd_{\theta}+2\nu,\\
A_{+} &=x\theta\pd_{x}+x\pd_{\theta}+2\nu\theta,
\end{align*}
where $\theta$ is a Grassmann variable $\theta^2=0$ commuting with $x$. It is easy to check that the operators satisfy the following relations
$$[A_0,A_{\pm}]=\pm A_{\pm} \qquad\text{ and}\qquad \{A_+,A_-\}=A_0,$$
where as usual $[a,b]=ab-ab$ denotes the commutator of the operators $a$ and $b$, and $\{a,b\}=ab+ba$ denotes the anticommutator. 
Therefore, the operators $\hat{A}_{\pm}=\sqrt{2}A_{\pm}$, $\hat{A}_{0}=A_{0}$ define a representation of the superalgebra $\osp$. Moreover
$$A_-^2=\pd_x.$$
Let $P$ be the operator acting on functions $f(x,\theta)$ by
$$P(f(x,\theta))=f(x,-\theta).$$
Clearly, $P$ is an involution, i.e. $P^2=1$ and 
$$[P,A_0]=0, \qquad \{P,A_{\pm}\}=0.$$
This operator $P$ is realized as follows $$P=1-2\theta \pd_{\theta}.$$
The Casimir operator 
$$Q=(A_0-2A_{+}A_{-}-1/2)P$$
commutes with $A_0$, $A_{\pm}$ and $P$. One can also easily check that in our current realization it holds that $$Q=2\nu-1/2.$$

\section{Three copies and the Bannai--Ito algebra}
For $j=1,2,3$, let
\begin{align*}
A^{(j)}_{-} &=\theta_j\pd_{x_j}+\pd_{\theta_j},\\
A^{(j)}_{0} &=2x_j\pd_{x_j}+\theta_j\pd_{\theta_j}+2\nu_j,\\
A^{(j)}_{+} &=x_j\theta_j\pd_{x_j}+x_j\pd_{\theta_j}+2\nu_j\theta_j,
\end{align*}
where $\{\theta_j\}$ are anticommuting, i.e.
$$\theta_i\theta_j=-\theta_j\theta_i.$$
Note that for $i\neq j$ we have
\begin{equation*}
[A^{(i)}_0, A^{(j)}_0]=0, \qquad [A^{(i)}_0, A^{(j)}_{\pm}]=0, \qquad \{A^{(i)}_{\pm}, A^{(j)}_{\pm}\}=0,\qquad \{A^{(i)}_{\pm}, A^{(j)}_{\mp}\}=0.
\end{equation*}
In particular, these commutativity relations imply that for every nonempty $S\subset\{1,2,3\}$ we obtain another representation of $\osp$ with operators defined by 
\begin{align*}
A^{(S)}_{0} =\sum_{i\in S} A^{(i)}_{0},\qquad A^{(S)}_{\pm} =\sum_{i\in S} A^{(i)}_{\pm}, \qquad P^{(S)}=\prod_{i\in S} P^{(i)},
\end{align*}
and associated Casimir operator
$$Q^{(S)}=(A^{(S)}_0-2A^{(S)}_{+}A^{(S)}_{-}-1/2)P^{(S)}.$$
Note that $Q^{(S)}$,  usually called intermediate Casimir operator, commutes with $A^{(T)}_0$, $A^{(T)}_{\pm}$ and $P^{(T)}$ for every $T\supset S$ and $Q^{(i)}=2\nu_i-1/2$.

The Casimir operators $Q^{(S)}$ provide a realisation of the Bannai--Ito algebra \textit{i.e.} they satisfy the following relations 
\begin{eqnarray}
{} \{Q^{(12)}, Q^{(23)} \}= Q^{(13)}  + 2Q^{(1)}Q^{(3)} + 2Q^{(2)}Q^{(123)}\,,\\ 
{} \{Q^{(12)}, Q^{(13)} \} = Q^{(23)}+ 2Q^{(2)}Q^{(3)} + 2Q^{(1)}Q^{(123)}\,,\\
{} \{Q^{(13)}, Q^{(23)} \} =Q^{(12)}+ 2Q^{(1)}Q^{(2)} + 2Q^{(3)}Q^{(123)}\,.
\end{eqnarray}  
In this realisation of the Bannai--Ito algebra, the following relation also holds
$$
( Q^{(12)})^2 +( Q^{(13)})^2 +( Q^{(23)})^2 +\frac{1}{4}=( Q^{(123)})^2 +( Q^{(1)})^2 +( Q^{(2)})^2 +( Q^{(3)})^2 \,.
$$

\section{Dimensional reduction}
Below we work with the superspace 
$$\Vset=\Cset[x_1,x_2,x_3]\langle \theta_1,\theta_2,\theta_3\rangle.$$
We can decompose $\Vset$ as the direct sum of the odd subspace $\Vseto$  and the even subspace $\Vsete$ as follows
\begin{equation}
\Vset=\Vseto\oplus \Vsete,
\end{equation}
where 
\begin{align*}
\Vseto&=\Cset[x_1,x_2,x_3] \theta_1\oplus \Cset[x_1,x_2,x_3] \theta_2\oplus \Cset[x_1,x_2,x_3] \theta_3\oplus \Cset[x_1,x_2,x_3] \theta_1\theta_2\theta_3\\
\Vsete&=\Cset[x_1,x_2,x_3] \oplus\Cset[x_1,x_2,x_3] \theta_1\theta_2 \oplus \Cset[x_1,x_2,x_3] \theta_1\theta_3 \oplus \Cset[x_1,x_2,x_3] \theta_2\theta_3.
\end{align*}
These subspaces can be characterized as the spaces consisting of skew-symmetric and symmetric functions with respect to the involution $P^{(123)}$ on $\Vset$.
Since $P^{(123)}$ anti commutes with $A^{(123)}_{-}$, we can decompose the $\ker(A^{(123)}_{-})$ on $\Vset$ as a direct sum of odd and even functions
$$\ker(A^{(123)}_{-})= (\ker(A^{(123)}_{-})\cap \Vseto)\oplus (\ker(A^{(123)}_{-})\cap \Vsete).$$
To describe each component, we fix 
$$u=x_1-x_2\qquad\text{ and }\qquad v=x_2-x_3,$$
and we consider the linear transformations 
\begin{align*}
&\cO_{1}, \cO_{2}:\Cset[u,v]\to \Vseto\\
&\cE_{1},\cE_{2}:\Cset[u,v]\to \Vsete
\end{align*}
defined as follows
\begin{align*}
\cO_{1}(h(u,v))&=h(u,v)(\theta_1-\theta_2)+h_{v}(u,v)\theta_1\theta_2\theta_3, \\
\cO_{2}(h(u,v))&=h(u,v)(\theta_2-\theta_3)-h_{u}(u,v)\theta_1\theta_2\theta_3, \\
\cE_{1}(h(u,v))&=h(u,v)(\theta_1\theta_2-\theta_1\theta_3+\theta_2\theta_3), \\
\cE_{2}(h(u,v))&=h(u,v)+h_u(u,v)\theta_1\theta_2+h_v(u,v)\theta_2\theta_3, 
\end{align*}
where $h(u,v)\in\Cset[u,v]$.

\begin{Proposition}\label{pr3.1}
Let $F\in\Vset$. Then $F$ solves the equation
\begin{equation}\label{eq1}
A^{(123)}_{-} F=0,
\end{equation}
if and only if $F$ can be written as
\begin{equation}\label{eq2}
F=F_o+F_e,
\end{equation}
where 
\begin{align}
F_o&=\cO_{1}(h_1(u,v))+\cO_{2}(h_2(u,v))\\
F_e&=\cE_{1}(g_1(u,v))+\cE_{2}(g_2(u,v)),
\end{align}
for some $h_1(u,v),h_2(u,v),g_1(u,v),g_2(u,v)\in\Cset[u,v]$.
\end{Proposition}
\begin{proof} 
Since 
$$(A^{(123)}_{-})^2 F=(\pd_{x_1}+\pd_{x_2}+\pd_{x_3})F=0,$$ 
we see that the $x$ dependence in $F$  is only through the variables $u$ and $v$, i.e. we can write $F$ as
$$F=\sum_{i_1,i_2,i_3=0}^{1}G_{i_1,i_2,i_3}(u,v)\theta_1^{i_1}\theta_2^{i_2}\theta_3^{i_3},$$
where $G_{i_1,i_2,i_3}(u,v)\in \Cset[u,v]$. Substituting this into \eqref{eq1} and equating the coefficients in the different powers of $\theta_1,\theta_2,\theta_3$ shows that $F$ will satisfy \eqref{eq1} if and only if the representation in \eqref{eq2} holds.
\end{proof}

\section{Action of the intermediate Casimir operators on $\ker(A^{(123)}_{-})$}
Since for every nonempty $S\subset\{1,2,3\}$ the Casimir operator $Q^{(S)}$ commutes with $A^{(123)}_{-}$, it follows that $Q^{(S)}$ preserves the space of solutions of equation \eqref{eq1}. 
Below we compute the action of the intermediate Casimir operators $Q^{(12)}$, $Q^{(13)}$ and $Q^{(23)}$ on the basis of solutions of \eqref{eq1} described in Proposition~\ref{pr3.1}.
For $i\neq j$ we set $\nu_{ij}=\nu_i+\nu_j$.

\subsection{Action of $Q^{(12)}$}
We have
\begin{align*}
Q^{(12)}\circ \cO_1=&- \cO_1\circ (2u\pd_u+2\nu_{12}+1/2) -\cO_2\circ(2u\pd_v),\\
Q^{(12)}\circ \cO_2=&\cO_1\circ (2u\pd_u+4\nu_{1}) +\cO_2\circ(2u\pd_u+2\nu_{12}-1/2),\\
Q^{(12)}\circ \cE_1=&- \cE_1\circ (2u\pd_u+2\nu_{12}+1/2) +\cE_2\circ(2u),\\
Q^{(12)}\circ \cE_2=&- \cE_1\circ (2u\pd_u\pd_v+4\nu_{1}\pd_v) +\cE_2\circ(2u\pd_u+2\nu_{12}-1/2).
\end{align*}

\subsection{Action of $Q^{(13)}$}
We have
\begin{align*}
Q^{(13)}\circ \cO_1=&- \cO_1\circ (2\nu_{1}-2\nu_{3}-1/2) -\cO_2\circ(2(u+v)\pd_v+4\nu_{3}),\\
Q^{(13)}\circ \cO_2=&- \cO_1\circ (2(u+v)\pd_u+4\nu_{1}) -\cO_2\circ(2\nu_{1}-2\nu_{3}+1/2),\\
Q^{(13)}\circ \cE_1=&- \cE_1\circ (2\nu_{13}-3/2) -\cE_2\circ(2u+2v),\\
Q^{(13)}\circ \cE_2=&- \cE_1\circ (2(u+v)\pd_u\pd_v+4\nu_{3}\pd_u+4\nu_{1}\pd_v) +\cE_2\circ(2\nu_{13}-1/2).
\end{align*}

\subsection{Action of $Q^{(23)}$}
We have
\begin{align*}
Q^{(23)}\circ \cO_1=&\cO_1\circ (2v\pd_v +2\nu_{23}-1/2) +\cO_2\circ(2v\pd_v+4\nu_{3}),\\
Q^{(23)}\circ \cO_2=&- \cO_1\circ (2v\pd_u) -\cO_2\circ(2v\pd_v +2\nu_{23}+1/2),\\
Q^{(23)}\circ \cE_1=&- \cE_1\circ (2v\pd_v +2\nu_{23}+1/2) +\cE_2\circ(2v),\\
Q^{(23)}\circ \cE_2=&- \cE_1\circ (2v\pd_u\pd_v+4\nu_{3}\pd_u) +\cE_2\circ(2v\pd_v +2\nu_{23}-1/2).
\end{align*}

\section{Diagonalization of $Q^{(123)}$ and $Q^{(12)}$ in the odd subspace}
\subsection{Spectral equations for $Q^{(123)}$}
Consider an odd element 

\begin{equation}\label{5.1}
f(\bfx;\bft)=\cO_1(h(u,v))+\cO_2(g(u,v))\in \ker A^{(123)}_{-},
\end{equation}
 where $\bfx=(x_1,x_2,x_3)$ and $\bft=(\theta_1,\theta_2,\theta_3)$. 
Then the spectral equation
\begin{equation}\label{Q123o}
Q^{(123)}f(\bfx;\bft)=\mu f(\bfx;\bft)
\end{equation}
is equivalent to the equations
\begin{align*}
uh_u+vh_v&=-(\nu_{123}+\mu/2+1/4)h,\\
ug_u+vg_v&=-(\nu_{123}+\mu/2+1/4)g,
\end{align*}
where $\nu_{123}=\nu_1+\nu_2+\nu_3$. The last two equations are satisfied if and only if $h$ and $g$ are homogeneous in $u$ and $v$ of the same degree $N$, where 
$$N=-(\nu_{123}+\mu/2+1/4),$$ 
or equivalently
$$\mu=-(2N+2\nu_{123}+1/2).$$
Since $Q^{(123)}$ commutes with each of the operators $Q^{(12)}$, $Q^{(13)}$, $Q^{(23)}$, we can look at their restrictions on the space of solutions of \eqref{Q123o}.
\subsection{Spectral equations for $Q^{(12)}$}
Consider now the equation
\begin{equation}\label{Q12o}
Q^{(12)}f(\bfx;\bft) =\la f(\bfx;\bft),
\end{equation}
where $\la\in\Cset$ and $f(\bfx;\bft)$ is the element in $\ker A^{(123)}_{-}$ in \eqref{5.1}.
The coefficient of $\theta_2$ shows that
\begin{equation}\label{5.3}
(1+2\la+4\nu_1-4\nu_2)g=(1+2\la+4\nu_{12})h+4uh_u-4uh_v.
\end{equation}
It is straightforward to check that if $g$ and $h$ are homogeneous polynomials in $u$ and $v$ of the same degree satisfying \eqref{Q12o}, then $1+2\la+4\nu_1-4\nu_2$ can be zero only when $\nu_1=\nu_2=0$. 
Thus for generic $\nu_1,\nu_2$ we can assume that $1+2\la+4\nu_1-4\nu_2\neq 0$ and therefore \eqref{5.3} determines $g$ uniquely from $h$ as follows
\begin{equation}\label{Q12_el_k}
g=\frac{(1+2\la+4\nu_{12})h+4uh_u-4uh_v}{1+2\la+4\nu_1-4\nu_2}.
\end{equation}
Substituting the last formula into \eqref{Q12o}, we see that \eqref{Q12o} holds if and only if $h(u,v)$ satisfies the equation
\begin{equation}\label{Q12_eq_h}
\big(u^2\pd_u^2-u^2\pd_u\pd_v+(2\nu_{12}+1)u\pd_u-(2\nu_1+1)u\pd_v+(\nu_{12}^2-(2\la+1)^2/16)\big)h=0.
\end{equation}
If $h(u,v)$ is homogeneous of degree $N$ we can write it as
\begin{equation}\label{}
h(u,v)=u^N\phi(v/u),
\end{equation}
and substituting the last equation in \eqref{Q12_eq_h}, we obtain the following equation for $\phi(z)$:
\begin{equation}\label{5.7}
\Big(z(1+z)\pd_z^2-((2\nu_{12}+2N-1)z+2\nu_1+N)\pd_z+(N+\nu_{12})^2-(2\la+1)^2/16\Big)\phi(z)=0.
\end{equation}
If the last equation has a solution which is a polynomial of degree $k\leq N$, then the coefficient of $z^k$ yields
$$(N-k+\nu_{12})^2=\frac{(2\la+1)^2}{16}.$$
This leads to 
\begin{equation}
\frac{2\la+1}{4}=\pm(N-k+\nu_{12}),\quad \text{ or equivalently }\quad\la=\pm2(N-k+\nu_{12})-\frac{1}{2}
\end{equation}
and for these values of $\la$, equation \eqref{5.7} reduces to 
\begin{equation}\label{5.9}
\Big(z(1+z)\pd_z^2-((2\nu_{12}+2N-1)z+2\nu_1+N)\pd_z+k(2N+2\nu_{12}-k)\Big)\phi(z)=0.
\end{equation}
Note that this equation is the hypergeometric differential equation with three regular singular points: $0$, $-1$ and $\infty$. For generic parameters $\nu_1,\nu_2$, up to a constant factor, 
it has a unique polynomial solution of degree $k$ given in terms of the Jacobi polynomial by 
\begin{equation}\label{5.10}
\phi(z)=P^{(-N-2\nu_1-1,-N-2\nu_2)}_k\left(1+2z\right).
\end{equation}
Let us recall that the Jacobi polynomials are expressed as follows in terms of the hypergeometric functions
\begin{equation}\label{JacobiH}
P^{(\alpha,\beta)}_k(x)=\frac{(\alpha+1)_k}{k!}{}_2F_1\left( \begin{matrix}-k,\, k+\alpha+\beta+1 \\ \alpha+1 \end{matrix} \,;\, \frac{1-x}{2} 
\right).
\end{equation}
Summarizing the above computations we obtain the following theorem describing the common eigenfunctions of the operators $Q^{(123)}$ and $Q^{(12)}$ which can be simultaneously diagonalized on $\ker (A^{(123)}_{-})\cap \Vseto$.
\begin{Theorem}
For $N\in\Nset_0$ and $k\in\{0,\dots,N\}$, let 
\begin{subequations}\label{5.13}
\begin{align}
h_{k,N}(u,v)&=u^N P^{(-N-2\nu_1-1,-N-2\nu_2)}_k\left(1+2\frac{v}{u}\right), \\
g_{k,N}^{+}(u,v)&=\frac{2N+2\nu_{12}-k}{N+2\nu_1-k}\ u^N P^{(-N-2\nu_1,-N-2\nu_2)}_k\left(1+2\frac{v}{u}\right),\label{5.13b} \\
g_{k,N}^{-}(u,v)&=- u^N P^{(-N-2\nu_1,-N-2\nu_2)}_{k-1}\left(1+2\frac{v}{u}\right).\label{5.13c}
\end{align}
\end{subequations}
Then 
$$\{f_{k,N}^{+}(\bfx;\bft),f_{k,N}^{-}(\bfx;\bft):N\in\Nset_0, 0\leq k\leq N \}$$ 
where
$$f_{k,N}^{\pm }(\bfx;\bft)=\cO_1(h_{k,N}(u,v))+\cO_2(g_{k,N}^{\pm}(u,v)),
$$
form a basis of $\ker (A^{(123)}_{-})\cap \Vseto$ and 
\begin{subequations}\label{5.14}
\begin{align}
Q^{(123)}\,f_{k,N}^{\pm }(\bfx;\bft)&=-(2N+2\nu_{123}+1/2)f_{k,N}^{\pm }(\bfx;\bft),\\
Q^{(12)}\,f_{k,N}^{\pm }(\bfx;\bft)&=(\pm2(N-k+\nu_{12})-1/2)f_{k,N}^{\pm }(\bfx;\bft).
\end{align}
\end{subequations}
\end{Theorem}
 \proof The result for $h_{k,N}(u,v)$ follows directly from the discussion preceding the theorem.
 Using this result for $h_{k,N}(u,v)$ and the eigenvalues $\la=\pm2(N-k+\nu_{12})-1/2$ in \eqref{Q12_el_k}, one gets the following expressions
  for $g_{k,N}^{\pm}(u,v)$:
 \begin{subequations}\label{eq:gg}
\begin{align*}
g_{k,N}^{+}(u,v)&=\frac{(N-k+2\nu_{12})h_{k,N}(u,v)+u(\pd_uh_{k,N}(u,v)-\pd_vh_{k,N}(u,v))}{N-k+2\nu_{1}},\\
g_{k,N}^{-}(u,v)&=\frac{(N-k)h_{k,N}(u,v)-u(\pd_uh_{k,N}(u,v)-\pd_vh_{k,N}(u,v))}{N-k+2\nu_{2}}.
\end{align*}
\end{subequations}
Note that 
$$\text{ if }\quad h=u^N\phi(v/u),\quad \text{ then }\quad(u\pd_u-u\pd_v)h=u^N[(N-(t+1)\pd_t)\phi(t)]|_{t=v/u}.$$
Using the above formula and the identity 
$$\left(-k+(t+1)\partial_t\right)P_k^{(\alpha,\beta)}(1+2t)
=(k+\beta)P_{k-1}^{(\alpha+1,\beta)}(1+2t)$$
we obtain the explicit formulas \eqref{5.13b}-\eqref{5.13c} for $g_{k,N}^{\pm}(u,v)$.
\endproof

Using the fact that the intermediate Casimir operators satisfy the Bannai--Ito algebra and knowing the spectrum of $Q^{(12)}$ and $Q^{(123)}$, one can deduce that the action of $Q^{(23)}$ is tridiagonal as follows
\begin{subequations}\label{eq:Q23o}
\begin{align}
 &Q^{(23)} f_{k,N}^{+}(\bfx;\bft)= \alpha^+_kf_{k,N}^{-}(\bfx;\bft) + \beta^+_{k} f_{k,N}^{+}(\bfx;\bft) +\gamma_k^+ f_{k+1,N}^{-}(\bfx;\bft),  \\
 &Q^{(23)} f_{k,N}^{-}(\bfx;\bft)= \alpha^-_{k-1}f_{k-1,N}^{+}(\bfx;\bft) + \beta^-_{k} f_{k,N}^{-}(\bfx;\bft) +\gamma_k^- f_{k,N}^{+}(\bfx;\bft).
\end{align}
\end{subequations}
The coefficients are also constrained by the algebra as follows:  
\begin{align}
 &\alpha^-_{k} \gamma_k^+ = \frac{4(k+1)(N-k)(2\nu_{12}+2N-k)(2\nu_{12}+N-k-1)}{(2\nu_{12}+2N-2k-1)^2},\label{ao}\\
&\alpha^+_{k} \gamma_k^- = \frac{(2\nu_1+N-k)(2\nu_2+N-k)(2\nu_3+k)(2\nu_{123}+2N-k)}{(\nu_{12}+N-k)^2},\label{bo}\\
&\beta^\pm_k= \pm \frac{(\nu_1-\nu_2)(\nu_{12}+2\nu_3+N)}{\nu_{12}+N-k} \mp \frac{(2\nu_{12}-1)(2\nu_{12}+2N+1)}{2(2\nu_{12}+2N-2k\mp 1)}.
\label{eq:betapm}
\end{align}
The exact expressions of these coefficients are given in the following corollary.
\begin{Corollary}
The coefficients in relations \eqref{eq:Q23o} are
\begin{eqnarray*}
&&\alpha_k^+=\frac{(2\nu_2+N-k)(2\nu_{123}+2N-k)}{\nu_{12}+N-k},\\
&&\alpha_k^-=\frac{2(2\nu_{1}+N-k+1)(2\nu_{12}+N-k)}{2\nu_{12}+2N-2k+1},\\
&&\gamma_k^+=\frac{2(N-k)(k+1)(2\nu_{12}+2N-k)}{(2\nu_1+N-k)(2\nu_{12}+2N-2k-1)},\\
&&\gamma_k^-=\frac{(2\nu_3+k)(2\nu_{1}+N-k)}{\nu_{12}+N-k}.
\end{eqnarray*}
\end{Corollary}
\proof The coefficients of $\theta_1$ in the relations \eqref{eq:Q23o}, computed by using the explicit formulas for $f_{k,N}^{\pm}(\bfx;\bft)$, provide constraints between the coefficients $\alpha^\pm$, $\beta^\pm$ and $\gamma^\pm$. Combing these constraints with the expression \eqref{eq:betapm} of $\beta^\pm_k$, one gets the expressions in the corollary.
\endproof
The values of these coefficients are compatible with the relations \eqref{ao}-\eqref{bo}.

\section{Diagonalization of $Q^{(123)}$ and $Q^{(12)}$ in the even subspace}
\subsection{Spectral equations for $Q^{(123)}$}
Consider an even element 
\begin{equation}\label{6.1}
f(\bfx;\bft)=\cE_1(h(u,v))+\cE_2(g(u,v))\in \ker A^{(123)}_{-}.
\end{equation}
Then the spectral equation
\begin{equation}\label{6.2}
Q^{(123)}f(\bfx;\bft)=\mu f(\bfx;\bft)
\end{equation}
is equivalent to the equations
\begin{align*}
uh_u+vh_v&=(-\nu_{123}+\mu/2-3/4)h,\\
ug_u+vg_v&=(-\nu_{123}+\mu/2+1/4)g.
\end{align*}
The last two equations are satisfied if and only if $g$ is homogeneous of degree $N$ and $h$ is homogeneous of degree $N-1$, where
$$N=-\nu_{123}+\mu/2+1/4.$$ 
or equivalently
$$\mu=2N+2\nu_{123}-1/2.$$
\subsection{Spectral equations for $Q^{(12)}$}
Consider now the equation
\begin{equation}\label{6.3}
Q^{(12)}f(\bfx;\bft) =\la f(\bfx;\bft),
\end{equation}
where $\la\in\Cset$ and $f(\bfx;\bft)$ is the element in $\ker A^{(123)}_{-}$ in \eqref{6.1}.
The free term (i.e. the coefficient of $\theta_1^0\theta_2^0\theta_3^0$) shows that
\begin{equation}\label{6.4}
h(u,v)=\frac{(1+2 \la -4\nu_{12})}{4 u}g(u,v)- \pd_u g(u,v).
\end{equation}
Substituting the last formula into \eqref{6.3}, we see that \eqref{6.3} holds if and only if $g(u,v)$ satisfies the equation
\begin{equation}\label{6.5}
\big(u^2\pd_u^2-u^2\pd_u\pd_v+2\nu_{12}u\pd_u-2\nu_1u\pd_v+\nu_{12}(\nu_{12}-1)+1/4-(2\la-1)^2/16\big)g=0.
\end{equation}
If $g(u,v)$ is homogeneous of degree $N$ we can write it as
\begin{equation}\label{}
g(u,v)=u^N\phi(v/u),
\end{equation}
and substituting the last formula in \eqref{6.5}, we obtain the following equation for $\phi(z)$:
\begin{align}
&\Big(z(1+z)\pd_z^2-(2(\nu_{12}+N-1)z+2\nu_1+N-1)\pd_z \nonumber\\
&\qquad\qquad\qquad+(N+\nu_{12}-1/2)^2-(2\la-1)^2/16\Big)\phi(z)=0.\label{6.7}
\end{align}
If the last equation has a solution $\phi(z)$ which is a polynomial of degree $k\leq N$, then the coefficient of $z^k$ yields
$$(N-k+\nu_{12}-1/2)^2=\frac{(2\la-1)^2}{16}.$$
Equivalently
\begin{equation}
\la=\pm2(N-k+\nu_{12}-1/2)+1/2
\end{equation}
and for these values of $\la$ equation \eqref{6.7} reduces to 
\begin{align}
&\Big(z(1+z)\pd_z^2-(2(\nu_{12}+N-1)z+2\nu_1+N-1)\pd_z \nonumber\\
&\qquad\qquad\qquad +k(2N+2\nu_{12}-k-1)\Big)\phi(z)=0.\label{6.9}
\end{align}

For generic parameters $\nu_1,\nu_2$ and up to a constant factor, this equation has a unique polynomial solution of degree $k$ given by 
\begin{equation}\label{6.10}
\phi(z)= P_k^{(-N-2\nu_1,-N-2\nu_2)}(1+2z),
\end{equation}
where $P_k$ is the Jacobi polynomial \eqref{JacobiH}.
The next theorem describes the common eigenfunctions of the commuting operators $Q^{(123)}$ and $Q^{(12)}$ which can be simultaneously diagonalized on $\ker (A^{(123)}_{-})\cap \Vsete$.

\begin{Theorem}
For $N\in\Nset_0$ and $k\in\{0,\dots,N\}$, let 
\begin{subequations}\label{6.12}
\begin{align}
g_{k,N}(u,v)&=u^N\, P_k^{(-N-2\nu_1,-N-2\nu_2)}\left(1+2\frac{v}{u}\right), \\
h_{k,N}^{+}(u,v)&=u^{N-1}(N+2\nu_1-k) 
P_{k-1}^{(-N-2\nu_1,-N-2\nu_2+1)}\left(1+2\frac{v}{u}\right),\\
h_{k,N}^{-}(u,v)&=(k+1-2N-2\nu_{12}) \,u^{N-1}
P_k^{(-N-2\nu_1,-N-2\nu_2+1)}\left(1+2\frac{v}{u}\right).
\end{align}
\end{subequations}
Then 
$$\{f_{k,N}^{+}(\bfx;\bft):N\in\Nset_0,\; 0\leq k\;\leq N \}\cup 
\{f_{k,N}^{-}(\bfx;\bft):N\in\Nset_0,\; 0\leq k< N \}$$ 
where
$$
f_{k,N}^{\pm }(\bfx;\bft)=\cE_1(h_{k,N}^{\pm}(u,v))+\cE_2(g_{k,N}(u,v)),
$$
form a basis of $\ker (A^{(123)}_{-})\cap \Vsete$ and 
\begin{subequations}\label{6.14}
\begin{align}
Q^{(123)}\,f_{k,N}^{\pm }(\bfx;\bft)&=(2N+2\nu_{123}-1/2)f_{k,N}^{\pm }(\bfx;\bft),\\
Q^{(12)}\,f_{k,N}^{\pm }(\bfx;\bft)&=(\pm2(N-k+\nu_{12}-1/2)+1/2)f_{k,N}^{\pm }(\bfx;\bft). \label{eq:la2}
\end{align}
\end{subequations}
\end{Theorem}
\proof  The result for $g_{k,N}(u,v)$ follows directly from the discussion preceding the theorem.
 Using the formula for $g_{k,N}(u,v)$ and the eigenvalues $\la=\pm2(N-k+\nu_{12}-1/2)+1/2$ in \eqref{6.4}, one gets the following expressions
  for $h_{k,N}^{\pm}(u,v)$:
  \begin{align}
h_{k,N}^{+}(u,v)&=\frac{N-k}{u}g_{k,N}(u,v)- \pd_u g_{k,N}(u,v),\\
h_{k,N}^{-}(u,v)&=-\frac{N-k+2\nu_{12}-1}{u}g_{k,N}(u,v)- \pd_u g_{k,N}(u,v).
\end{align}
Note that 
$$\text{ if }\quad g=u^N\phi(v/u),\quad \text{ then }\quad u\pd_u g=u^N[(N-t\pd_t)\phi(t)]|_{t=v/u}.$$
Using the above formula and the identity 
$$(k-t \pd_t)P_k^{(\alpha,\beta)}(1+2t) =(\alpha+k)P_{k-1}^{(\alpha,\beta+1)}(1+2t)$$
we can write explicit formulas for $h_{k,N}^{\pm}(u,v)$ in terms of the Jacobi polynomials as stated in the theorem.
\endproof

Similarly to the arguments in the odd space, combining the fact that the intermediate Casimir operators satisfy the Bannai--Ito algebra with the spectrum of $Q^{(12)}$ and $Q^{(123)}$, one can deduce that the action of $Q^{(23)}$ is tridiagonal as follows
\begin{subequations}\label{eq:Q23e}
\begin{align}
 &Q^{(23)} f_{k,N}^{+}(\bfx;\bft)= \alpha^+_{k} f_{k-1,N}^{-}(\bfx;\bft) + \beta^+_{k} f_{k,N}^{+}(\bfx;\bft) +\gamma_k^+ f_{k,N}^{-}(\bfx;\bft),  \\
 &Q^{(23)} f_{k,N}^{-}(\bfx;\bft)= \alpha^-_kf_{k,N}^{+}(\bfx;\bft) + \beta^-_{k+1} f_{k,N}^{-}(\bfx;\bft) +\gamma_{k+1}^- f_{k+1,N}^{+}(\bfx;\bft).
\end{align}
\end{subequations}
The coefficients are also constrained by the algebra as follows:  
\begin{align}
 &\alpha^+_{k} \gamma_k^+ =\frac{k(2\nu_1+N-k)(2\nu_2+N-k)(2\nu_{12}+2N-k)}{(\nu_{12}+N-k)^2},\label{ae}\\
&\alpha^-_{k} \gamma_{k}^+ = \frac{4(2\nu_3+k)(N-k)(2\nu_{12}+N-k-1)(2\nu_{123}+2N-k-1)}{(2\nu_{12}+2N-2k-1)^2},\label{be}\\
&\beta^\pm_k= \mp \frac{(\nu_1-\nu_2)(\nu_{12}+N)}{\nu_{12}+N-k} \pm \frac{(2\nu_{12}-1)(2\nu_{12}+4\nu_3+2N+1)}{2(2\nu_{12}+2N-2k \mp1)}.
\label{eq:betapme}
\end{align}
The exact expressions of these coefficients are given in the following corollary.
\begin{Corollary}
The coefficients in relations \eqref{eq:Q23e} are
\begin{eqnarray*}
&&\alpha_k^+=\frac{(2\nu_1+N-k)(2\nu_2+N-k)}{\nu_{12}+N-k},\\
&&\alpha_k^-= \frac{2(2\nu_{12}+N-k-1)(2\nu_{123}+2N-k-1)}{2\nu_{12}+2N-2k-1},\\
&&\gamma_k^+= \frac{2(2\nu_3+k)(N-k)}{2\nu_{12}+2N-2k-1},\\
&&\gamma_k^-=\frac{k(2\nu_{12}+2N-k)}{\nu_{12}+N-k}.
\end{eqnarray*}
\end{Corollary}
\proof The coefficients without $\theta$'s in the relations \eqref{eq:Q23e}, computed by using the explicit expressions for $f_{k,N}^{\pm}(\bfx;\bft)$, provide constraints between the coefficients $\alpha^\pm$, $\beta^\pm$ and $\gamma^\pm$. Combining these constraints with formula \eqref{eq:betapme} for $\beta^\pm_k$, one gets the expressions stated in the corollary.
\endproof
The values of these coefficients are compatible with the relations \eqref{ae}-\eqref{be}.

\section{Conclusion}

Summing up, we have provided a model for the Bannai--Ito algebra on a superspace with $\mathbb{C}^2$ as body and with soul generated by three anticommuting Grassmann variables. The even and odd basis vectors were found to each have two components and to be realized in terms of Jacobi polynomials. The tridiagonal action of one Bannai--Ito generator in the eigenbasis of the other was explicitly calculated. The formulas thus obtained reflect contiguous relations of the Jacobi polynomials. 

It is worth mentioning that the Bannai-Ito algebra is isomorphic to the degenerate double affine Hecke algebra (DAHA) \cite{genest2016non} which has thus been modelled here in superspace by the same token.

For the sake of completeness, let us mention that an embedding of the Bannai--Ito algebra into $\osp$ was presented in \cite{baseilhac2018embedding} where an analytic realization of $\osp$ in terms of Dunkl operators was used to connect to some $-1$ - polynomials. It should be stressed that the Bannai--Ito model that results from combining the realization of $\osp$ used throughout in this paper with this embedding is clearly not the one that is attained from the dimensional reduction performed here.

Finally, it would definitely be of interest to develop along the lines pursued here realizations on superspaces of the higher rank Bannai--Ito algebras that have been constructed from the intermediate Casimir operators arising in manifold tensor products of $\osp$ \cite{de2016z2n}.

\section*{Acknowledgments}
The authors wish to acknowledge the enlightening discussions they had with Wouter van de Vijver in the early stages of the project. NC and PI are grateful to the Centre de Recherches Math\'ematiques (CRM) for supporting visits during which the research reported here was carried out.  NC is funded by the international research project AAPT of the CNRS and the ANR Project AHA ANR-18-CE40-0001. HDB is supported by the EOS-FWO Research Project grant 30889451.
PI is supported in part by Simons Foundation Grant \#635462. LV is supported in part by a Discovery Grant from NSERC.


\bibliographystyle{unsrt} 
\bibliography{ref_superspace.bib} 

\end{document}